\documentclass[11pt]{article}

\usepackage{amsmath,amsfonts,amsthm,amssymb}
\usepackage{xcolor}

\swapnumbers 

\newtheorem{lemma}{Lemma}[section]
\newtheorem{corollary}[lemma]{Corollary}
\newtheorem{theorem}[lemma]{Theorem}
\newtheorem{prop}[lemma]{Proposition}

\theoremstyle{definition} 
\newtheorem{definition}[lemma]{Definition}

\newtheorem{example}[lemma]{Example}
\newtheorem{remark}[lemma]{Remark}

\begin{document}

\title{Spectrality in convex sequential effect algebras}
    \author{Anna Jen\v cov\'a and
 Sylvia Pulmannov\'{a}{\footnote{ Mathematical Institute, Slovak Academy of
Sciences, \v Stef\'anikova 49, SK-814 73 Bratislava, Slovakia;
jenca@mat.savba.sk, pulmann@mat.savba.sk. }}}

\date{}

\maketitle

\abstract{For convex and sequential effect algebras, we study spectrality in the sense of
Foulis. We show that under additional conditions (strong archimedeanity, closedness in
norm and a certain monotonicity property of the sequential product), such effect algebra
is spectral if and only if every maximal commutative subalgebra is monotone
$\sigma$-complete. Two  previous results on existence of spectral
resolutions in this setting are shown  to require  stronger assumptions.
}

\medskip


\section{Introduction}

Effect algebras were introduced by Foulis and Bennett \cite{FoBe} as an algebraic
abstraction of the set of Hilbert space effects, that is, operators on a Hilbert space
lying in the interval between zero and the identity operator. The effects play an important role in  the mathematical description of quantum theory, since they represent  the yes-no measurements in quantum mechanics.
An effect algebra is called convex if it has a convex structure, \cite{GuPu, GPBB}. It was proved that any convex effect algebra
can be represented as an interval in an ordered vector space \cite{GPBB}, and under
additional conditions as the unit interval in an order unit space \cite{GPBB}.
Gudder and Greechie \cite{GuGr} introduced an additional operation of a
sequential product which is an analogue of the operation $(a,b)\mapsto a^{1/2}ba^{1/2}$
for Hilbert space effects $a,b$ and is interpreted as a description of a sequential
measurement. Effect algebras endowed with such a product are called sequential.

One of the important properties of Hilbert space effects is the existence of spectral
resolutions, which means that every effect can be expressed in terms of sharp effects representing sharp measurements. The sharp effects are precisely the projection operators.
Spectrality appears  as a crucial property in operational derivations of quantum theory,
see e.g. \cite{chiribella, Hardy, www, wetering}. It is therefore important to study possible notions of
  spectrality in some classes of effect algebras and to determine the properties and additional
  structures that ensure the existence of some type of spectral resolutions.

Perhaps the most well known extension of spectrality to order unit spaces is due to Alfsen
and Schultz \cite{AS, AlSh}. Their notion of spectral duality is based upon the geometry of dual order unit and base normed
spaces. A more algebraic definition was introduced in \cite{FPspectres}, following the
works by Foulis on spectrality in partially ordered unital abelian groups \cite{Fcomgroup,
Fgc, Funig, Forc}. The two approaches were compared in  \cite{JP1}. In both
definitions, a crucial role is played by compressions, generalizing the map $a\mapsto pap$
for a Hilbert space effect $a$ and a projection $p$. In particular, if there exists a suitable set of compressions with specified properties, each element has a
unique spectral resolution, or a \emph{rational} spectral resolution with values restricted to $\mathbb Q$ in the
case of partially ordered abelian groups, analogous to the spectral resolution of self-adjoint elements in von
Neumann algebras.

Following the ideas in \cite{Fcompog}, compressions and compression bases in effect algebras were studied by Gudder
\cite{Gu, Gudcomprba} and Pulmannov{\'a} \cite{Pucompr}. In \cite{JP2}, we proved that under the conditions of spectrality
specified for an effect algebra in \cite{Pucompr}, there exists a binary spectral
resolution, restricted to dyadic rationals, characterized by properties analogous to
spectral resolutions for  Hilbert space effects.

In the present work, we will concentrate on the special class of effect algebras that are both convex and sequential.
 Spectral resolutions in this setting were studied in \cite{Gucs}, where it was further assumed that any element is a finite
combination of indecomposable sharp elements summing up to identity, such collections of
sharp elements are called contexts. In \cite{Wet}, it was shown that if the effect
algebra is also monotone $\sigma$-complete, then each element has a spectral resolution,
in the sense that it can be written as a supremum and norm limit of simple elements, that
is, finite combinations of orthogonal sharp effects.

These works do not explicitly use
any compressions, but note that for sequential effect algebras,  there is a distinguished set of compressions given  by the
sequential product with a sharp element. Moreover, such compressions form a compression
base, \cite{Gudcomprba}. It is therefore natural to  study spectrality of convex and
sequential effect algebras in the sense derived from the works  of Foulis
(as in \cite{Pucompr,JP2}). This is precisely the aim of the present paper.
We show that under additional assumptions (strong archimedeanity, norm completeness
and a certain monotonicity property called the A-property), the effect algebra is spectral if and only
if every maximal commutative subalgebra is monotone $\sigma$-complete. We also show that
the conditions in both \cite{Gucs} and \cite{Wet} imply spectrality in the Foulis
sense.

After a preliminary section (Sec. \ref{sec:eas}) on general effect algebras, we describe
the notion of spectrality in Sec. \ref{sec:spectral}. Sec. \ref{sec:coseq} briefly
describes the convex and sequential effect algebras, Sec. \ref{sec:main} contains our main
results.

\section{Effect algebras}\label{sec:eas}

An \emph{effect algebra} \cite{FoBe} is a system $(E; \oplus,0,1)$ where $E$ is a nonempty set, $\oplus$ is a partially defined binary operation on $E$, and $0$ and $1$ are constants, such that the following conditions are satisfied:

\begin{enumerate}
\item[(E1)] If $a\oplus b$ is defined then  $b\oplus a$ is defined and $a\oplus b=b\oplus a$.
\item[(E2)] If $a\oplus b$ and $(a\oplus b)\oplus c$ are defined then  $b\oplus c$ and $a\oplus(b\oplus c)$ are defined and
$(a\oplus b)\oplus c=a\oplus(b\oplus c)$.
\item[(E3)] For every $a\in E$ there is a unique $a^\perp \in E$ such that $a\oplus a^\perp=1$.
\item[(E4)]If  $a\oplus 1$ is defined then $a=0$.
\end{enumerate}
Elements of $E$ are called \emph{effects}.
We write $a\perp b$ and say that $a$ and $b$ are \emph{orthogonal} if $a\oplus b$ exists.
In what follows, when we write $a\oplus b$, we tacitly assume that $a\perp b$. A partial
order is introduced on $E$ by defining $a\leq b$ if there is $c\in E$ with $a\oplus c=b$.
If such an element $c$ exists, it is unique, and we define $b\ominus a:=c$.  With respect
to this partial order we have $0\leq a\leq 1$ for all $a\in E$. The element
$a^\perp=1\ominus a$ in (E3)  is called the \emph{orthosupplement} of $a$. It can be shown
that $a\perp b$ iff $a\leq b^\perp$ (equivalently, $b\leq a^\perp$). Moreover $a\leq b$
iff $b^\perp \leq a^\perp$, and $a^{\perp\perp}=a$.

An element $a\in E$ is called \emph{sharp} if $a\wedge a^\perp=0$ (i.e., $x\leq a,a^\perp \implies x=0$). We denote the set of all sharp elements of $E$ by $E_S$. An element $a\in E$ is \emph{principal} if $x,y\leq a$, and $x\perp y$ implies that $x\oplus y\leq a$. It is easy to see that a principal element is sharp.

The algebra of Hilbert space effects described below is a prototypical example of an effect algebra on
which the above abstract definition is modelled.

\begin{example}\label{ex:effects} Let $\mathcal H$ be a Hilbert space and let $E(\mathcal H)$ be the set of operators on $\mathcal H$ such
 that $0\le A\le I$. For  $A,B\in E(\mathcal H)$, put $A\oplus B=A+B$ if $A+B\le I$,
otherwise $A\oplus B$ is not defined. Then $(E(\mathcal H);\oplus, 0, I)$ is an effect
algebra. Note that any sharp element
is principal and the set of sharp effects $E(\mathcal H)_S$ coincides with the set
$P(\mathcal H)$ of
projection operators on $\mathcal H$, that is, linear operators $p:\mathcal H\to \mathcal
H$ such that  $p=p^*=p^2$.

\end{example}

The effect algebra $E(\mathcal H)$ belongs to a larger class of effect algebras obtained
as intervals in partially ordered groups.

\begin{example}\label{ex:intervalea} Let $(G,u)$ be a partially ordered abelian group with an order unit $u$. Let $G[0,u]$ be the unit interval in $G$ (we will often write $[0,u]$ if the group $G$ is clear). For $a,b\in G[0,u]$, let  $a\oplus b$ be defined if $a+b\le u$ and in this case $a\oplus b=a+b$. It is easily checked that $(G[0,u],\oplus, 0,u)$ is an effect algebra. Effect algebras of this form are called \emph{interval} effect algebras. In particular, the real unit interval $\mathbb R[0,1]$ can be given a structure of an effect algebra. Note also that the Hilbert space effects in Example \ref{ex:effects} form an interval effect algebra.

\end{example}

By recurrence, the operation $\oplus$ can be extended to finite sums $a_1\oplus a_2\oplus \cdots \oplus a_n$ of (not necessarily different) elements $a_1,a_2,\ldots a_n$ of $E$.
If $a_1=\dots=a_n=a$ and $\oplus_i a_i$ exist, we write $\oplus_i a_i=na$. An effect
algebra $E$ is \emph{archimedean} if for $a\in E$, $na\le 1$ for all $n\in \mathbb N$ implies that $a=0$.

An infinite family $(a_i)_{i\in I}$ of elements of $E$ is called \emph{orthogonal} if every its finite subfamily has an $\oplus$-sum  in $E$.
If the element  $\oplus_{i\in I} a_i=\bigvee_{F\subseteq I}\oplus_{i\in F}a_i$ exists, where the supremum is taken over all finite subsets of $I$ exists, it is called the \emph{orthosum} of
the family $(a_i)_{i\in I}$. An effect algebra $E$ is a $\sigma$-\emph{effect algebra} if
it is \emph{$\sigma$-orthocomplete}, that is,  if the orthosum exists for any
$\sigma$-finite orthogonal subfamily of $E$. Equivalently, $E$ is \emph{monotone
$\sigma$-complete}, that is,  every ascending sequence $(a_i)_{i\in {\mathbb N}}$  has a supremum $a=\bigvee_i a_i$
   in $E$ (or every descending sequence $(b_i)_{i\in {\mathbb N}}$ has an infimum
   $b=\bigwedge_i b_i$)  in $E$. Equivalence of these two conditions was proved in  \cite{JePu}.


   A subset $F$ of $E$ is \emph{sup/inf -closed in E} if whenever  $M\subseteq F$ and $\wedge M$ ($\vee M$) exists in $E$, then $\wedge M \in F$ ($\vee M \in F$).

If $E$ and $F$ are effect algebras, a mapping $\phi: E\to F$ is a \emph{morphism} if it is
\emph{additive}: $a\perp b$ implies $\phi(a) \perp \phi(b)$ and  $\phi(a\oplus
b)=\phi(a)\oplus \phi(b)$,  and  $\phi(1)=1$. If $\phi :E\to F$ is a morphism, and $\phi(a)\perp \phi(b)$ implies $a\perp b$, then $\Phi$ is a \emph{monomorphism}. A surjective monomorphism is an \emph{isomorphism}.

A \emph{state} on an effect algebra $E$ is a morphism $s$  from $E$ into the effect
algebra ${\mathbb R}[0,1]$ (see Example \ref{ex:intervalea}). We denote the set of states on $E$ by $S(E)$. We say that
$S\subset S(E)$ is \emph{separating} if $s(a)= s(b)$ for every $s\in S$ implies that $a=
b$, and  \emph{ordering} (or order determining) if
$s(a)\leq s(b)$ for all $s\in S$ implies $a\leq b$. If $S$ is ordering, then it is separating, the converse does not hold.

A lattice ordered effect algebra $M$, in which  $(a\vee b)\ominus a=a\ominus (a\wedge b)$ holds for all $a,b\in M$, is called an \emph{MV-effect algebra}.
We recall that MV-effect algebras are equivalent with MV-algebras, which were introduced
by \cite{Chang} as an algebraic basis for many-valued logic. It was proved in \cite{Mun}
that MV-algebras are equivalent to  lattice ordered groups with order unit, in the sense
of category theory.

\section{Spectrality in effect algebras}\label{sec:spectral}

\subsection{Compressions on  effect algebras }

The next  definition follows the works of Foulis {\cite{Fcompog}}, Gudder \cite{Gu} and Pulmannov\'a \cite{Pucompr}.

\begin{definition}\label{de:compr} Let $E$ be an effect algebra.
\begin{enumerate}
\item An additive map $J: E\to E$ is a
\emph{retraction} if $a\leq J(1)$ implies $J(a)=a$. The element $p:=J(1)$ is called the \emph{focus} of $J$.
\item A retraction with focus $p$ is a \emph{compression} if $J(a)=0\ \Leftrightarrow\
a\leq p^\perp$.
\item If $I$ and $J$ are retractions we say that $I$ is a \emph{supplement of} $J$ if $\ker(J)=I(E)$ and $\ker(I)=J(E)$.
\end{enumerate}
\end{definition}

It is easily seen that any retraction is idempotent. The focus of a retraction $J$ is a principal element and we have
$J(E)=[0,p]$,
moreover, $J$ is a compression if and only if $\mathrm{Ker}(J)=[0,p^\perp]$. If a retraction $J$ has a supplement $I$, then both $I$ and $J$ are
compressions and $I(1)=J(1)^\perp$.   For these and further properties see \cite{Gudcomprba, Pucompr}.

\begin{example}\label{ex:compr_effects} Let $E(\mathcal H)$ be the algebra of effects on $\mathcal H$ (Example
\ref{ex:effects}) and let $p\in E(\mathcal H)$ be a projection. Let us define the map $J_p:a\mapsto pap$, then $J_p$ is a compression on
$E(\mathcal H)$ and $J_{p^\perp}$ is a supplement of $J_p$. By \cite{Fcompog}, any retraction on $E(\mathcal H)$
is of this form for some projection $p$. In particular, any projection is the focus of a unique retraction $J_p$ with a
(unique) supplement $J_{p^\perp}$. Effect algebras such that any retraction is supplemented  and  uniquely determined by its
focus are called \emph{compressible}, \cite{Gu,Fcomgroup}.
\end{example}

Recall that two elements $a,b\in E$ are (Mackey) \emph{compatible} if there are elements $a_1,b_1, c\in E$
such that $a_1\oplus b_1\oplus c$ exists and $a=a_1\oplus c, b=b_1\oplus c$. In this case
we shall write $a\leftrightarrow b$. If $F\subseteq E$ and $a,b \in F$, we say that $a,b$ are
\emph{compatible in F} if
$a\leftrightarrow b$ and the elements $a_1,b_1,c$ can be chosen in $F$. It was proved in
\cite{Gudcomprba} that this is equivalent to compatibility in $E$ if  $F$ is a \emph{normal} sub-effect
algebra: for all $e,f,d\in E$ such that  $e\oplus f\oplus d$ exists in $E$, we have $e\oplus d, f\oplus d\in F\ \implies\ d\in F$.

\begin{definition}\label{de:comprbase} {\rm \cite{Gudcomprba} }  A family $(J_p)_{p\in P}$ of compressions on an effect algebra $E$ indexed by a  sub-effect algebra $P$ of $E$  is called a  \emph{compression base} on  $E$ if the following conditions hold:
\begin{enumerate}

\item[(C1)] each $p\in P$ is the focus of $J_p$,
\item[(C2)] $P$ is normal,
\item[(C3)] if $p,q,r\in P$ and $p\oplus q\oplus r$ exists, then $J_{p\oplus q}\circ J_{q\oplus r}=J_q$.
\end{enumerate}
Elements of $P$ are  called \emph{projections}.

\end{definition}

\begin{example}\label{ex:cb_hilb} Let $P(\mathcal H)$ be the set of all projections on a Hilbert space $\mathcal H$ and let $J_p$ for
$p\in P(\mathcal H)$ be as
in Example \ref{ex:compr_effects}. It is easily observed that the set $(J_p)_{p\in P(\mathcal H)}$ is a compression base
in $E(\mathcal H)$, moreover, it is the unique compression base in $E(\mathcal H)$ which
is \emph{maximal} in the sense that it is not contained in any other compression base.
More generally, the set of all compressions in a compressible effect algebra is the unique
maximal compression base, \cite{Gudcomprba}.

\end{example}

It was proved in \cite{JP2} that an equivalent definition of a compression base  is obtained if
only (C1) is required along with the condition
\begin{enumerate}
\item[(C2')] if $p,q\in P$ and $p\leftrightarrow q$ (in $E$), then $J_p\circ J_q=J_r$ for
some $r\in P$.
\end{enumerate}
This definition  is perhaps more clearly motivated by
analogy with Example \ref{ex:cb_hilb}, since it  corresponds to the fact that for the Hilbert space effect algebra $E(\mathcal H)$
and projections $p,q\in P(\mathcal H)$, we have $p\leftrightarrow q$ iff $pq\in
P(\mathcal H)$ and then
$J_p\circ J_q=J_{pq}$.


By  \cite[Corollary 4.5]{Gu} and \cite[Theorem 2.1]{Pucompr},  the set $P$ as a subalgebra
of $E$ is a \emph{regular}
orthomodular poset (OMP) with the orthocomplementation $a\mapsto a^\perp$, and
$J_{p^\perp}$ is a supplement of $J_p$. Recall
that an OMP $P$ is {regular} if for all $a,b,c\in P$, if $a,b$ and $c$ are pairwise compatible, then
$a\leftrightarrow b\vee c$ and $a\leftrightarrow b\wedge c$, \cite{PtPu, Harding}.

\begin{example}\label{ex:compr_mv} Let $M$ be an MV-effect algebra. Every retraction on
$M$ is of the form $U_p(a)=p\wedge a$ for some $p\in M_S$ as its focus.  Moreover, $M$ is a
compressible effect algebra and the set $(U_p)_{p\in M_S}$ is the unique  maximal compression base on $M$ (cf.  \cite[Theorem 3.1]{Pucompr}).
\end{example}

\subsection{Compatibility and commutants}
\label{sec:commutants}

From now on, we will assume that $E$ is an effect algebra with a  fixed compression base $(J_p)_{p\in P}$.
By \cite[Lemma 4.1]{Pucompr} we have the following.

\begin{lemma}\label{le:comE} If $p\in P, a\in E$, then the following statements are equivalent:
\begin{enumerate}
\item $J_p(a)\leq a$,
\item $a=J_p(a)\oplus J_{p^\perp}(a)$,
\item $a\in E[0,p]\oplus E[0,p^\perp]$,
\item $a \leftrightarrow p$,
\item $J_p(a)=p\wedge a$.
\end{enumerate}
\end{lemma}

The \emph{commutant} of $p$ in $E$ is defined by
\[
C(p):=\{ a\in E: a=J_p(a) \oplus J_{p^\perp}(a)\}.
\]
If $Q\subseteq P$, we write $C(Q):=\bigcap_{p\in Q}C(p)$. Similarly, for an element $a\in E$, and a subset $A\subseteq E$, we write
\[
PC(a):=\{p\in P,\ a\in C(p)\},\qquad PC(A):=\bigcap_{a\in A} PC(a).
\]
We also define
\[
CPC(a):= C(PC(a)),\quad P(a):=CPC(a)\cap PC(a)=PC(PC(a)\cup\{a\}).
\]
The set $P(a)\subseteq P$ will be called the \emph{P-bicommutant} of $a$ (that is, $P(a)$ is
the set of all projections $p\in P$ which are compatible with  $a$ and with all
projections compatible with $a$). For a subset $Q\subseteq E$, we put
\[
P(Q):= PC(PC(Q)\cup Q).
\]
Note that the elements in $P(Q)$ are pairwise
compatible and since $P$ is a regular OMP, this implies that $P(Q)$ is a Boolean subalgebra in $P$.

\begin{lemma} \label{lemma:compatible_projs} Let $p,q\in P$, $a\in E$.
\begin{enumerate}
\item \cite[Lemma 4.2]{Pucompr} If $p\perp q$ and either $a\in C(p)$ or $a\in C(q)$, then
\[
J_{p\vee q}(a)=J_{p\oplus q}(a)=J_p(a)\oplus J_q(a).
\]
\item\cite[Cor. 4.3]{Gudcomprba} If $p\leftrightarrow q$, then $J_pJ_q=J_qJ_p=J_{p\wedge q}$.

\end{enumerate}

\end{lemma}

Recall that a maximal set of pairwise compatible elements in a regular  OMP $P$ is called a \emph{block} of $P$ \cite[Corollary 1.3.2]{PtPu}.
It is well known that every block  $B$ is a Boolean subalgebra of $P$ \cite[Theorem 1.3.29]{PtPu}.
If $B$ is a block of $P$, the set $C(B)$ will be called a \emph{C-block} of $E$.

\begin{example} \label{ex:hilb_commutant}
 Let  $E=E(\mathcal H)$ for a Hilbert space $\mathcal H$. It is easily checked that for any projection $p\in P(\mathcal H)$,
\[
C(p)=\{p\}'\cap E=\{a\in E, pa=ap\}
\]
and for any $a\in E$,
\[
P(a)=\{a\}''\cap P(\mathcal H)
\]
(here $C'$ denotes the usual commutant of a subset of bounded operators $C\subset B(\mathcal H)$). The C-blocks are the unit intervals in maximal abelian von Neumann subalgebras of $B(\mathcal H)$.
\end{example}

\subsection{Spectral effect algebras}

In this section we recall the definition of a spectral  effect algebra, introduced in
\cite{Pucompr}. Remember that $E$ is  an effect algebra with a distinguished compression base
$(J_p)_{p\in P}$. Spectrality is defined by two properties of the compression base. The
first property is an analogue of existence of support projections.

\begin{definition}\label{de:projcov} If $a\in E$ and $p\in P$, then $p$ is a \emph{projection cover} for $a$ if, for all
$q\in P$, $a\leq q \,\Leftrightarrow \, p\leq q$. We say that $E$  has the \emph{projection cover property} if every effect $a\in E$ has a (necessarily unique) projection cover. The projection cover of $a\in E$ will be denoted as $a^\circ$.

\end{definition}

For an element $a\in E$, we may also define the \emph{floor} of $a$ as the largest projection under
$a$ (if it exists). It will be denoted by $a_\circ$. The relation to the projection cover
is $(a^\perp)_\circ =(a^\circ)^\perp$, this is rather obvious from $p\le a \iff a^\perp\le
p^\perp$. It follows that  $E$ has the projection cover property if and only if  any element has
the floor.

\begin{theorem}[{\cite[Thm. 5.2]{Gudcomprba}, \cite[Thm. 5.1]{Pucompr}}]\label{th:projcovoml}  Suppose that $E$ has the projection cover property.
Then $P$ is an orthomodular
lattice (OML). Moreover, $\mathcal P$ is sup/inf-closed in $E$.
\end{theorem}

\begin{prop}\label{prop:pc_commutant} Let $E$ have the projection cover property. Then for any  $a\in E$, $a^\circ\in P(a)$.

\end{prop}

\begin{proof}  Since $a\le a^\circ$, $a\in C(a^\circ)$ by Lemma \ref{le:comE} (iv). The rest follows by \cite[Thm. 5.2 (i)]{Pucompr}.

\end{proof}

The second property, the b-comparability,  was introduced  as
an analogue of the general comparability property in unital partially ordered abelian groups
\cite{Funig}, where it can be interpreted  as  existence of orthogonal decompositions of
elements into a positive and a negative part.

In the case of effect algebras with compression bases, the definition is  more
involved. We first introduce a notion analogous to commutativity of Hilbert space effects.
Recall that
for  $a,b\in E(\mathcal H)$, $ab=ba$ implies that $a\leftrightarrow b$, but the converse
is not necessarily true unless $a$ or $b$ is a projection. To obtain the corresponding
notion for an effect algebra $E$ with a compression base, we will need a further
property.

\begin{definition}{\cite[Definition 6.1]{Pucompr}}\label{de:belement}  We will say that $a\in E$ has the
\emph{b-property} (or is a \emph{b-element}) if there is a Boolean subalgebra $B(a)\subseteq P$ such that for all $p\in P$, $a\in C(p) \,\Leftrightarrow \, B(a)\subseteq C(p)$. We say that $E$ has the \emph{b-property} if every $a\in E$ is a b-element.
\end{definition}

The Boolean subalgebra  $B(a)$ in the above definition is in general not unique. By
\cite[Lemma 3.20]{JP2}, the bicommutant $P(a)$ of $a$ is the largest such subalgebra.
Further, by \cite[Proposition 6.1]{Pucompr}, every projection $q\in P$ is a b-element with
$B(q)=\{0,q,q^\perp,1\}$ and if  an element $a\in E$ is a b-element,
then there is a block $B$ of $P$ such that $a\in C(B)$.

Let $e,f\in E$ have the b-property. We say that $e$ and $f$ \emph{commute}, in notation $eCf$, if
\begin{equation}\label{eq:commut}\ P(e) \leftrightarrow P(f).
\end{equation}
By \cite[Lemma 2.18]{JP2}, this is equivalent to $B(e)\leftrightarrow B(f)$ for any choice
of the Boolean subalgebras $B(e)$ and $B(f)$. For $p\in P$ we have $eCp \iff e\leftrightarrow p \iff e\in C(p)$,
\cite[Lemma 6.1]{Pucompr}, so this
definition coincides with compatibility if one of the elements is a projection.

\begin{remark}
Roughly speaking, the b-property can be seen as the requirement that there are `enough'
projections in $E$. Of course, this depends on the choice of the compression base. For
example, if the compression base is trivial,  that is, $P=\{0,1\}$, then $E$ trivially has
the b-property and $aCb$ for any $a,b\in E$.
\end{remark}

\begin{theorem}\label{thm:cblock}{\rm \cite[Theorem 2.19]{JP2}} Assume that  $E$ has the
b-property.  Then $E$ is covered by its C-blocks. Moreover,
C-blocks in $E$ coincide with maximal sets of pairwise commuting elements in $E$.

\end{theorem}

\begin{definition}\label{de:b-compar} {\rm \cite[Definition 6.3]{Pucompr} } An effect algebra  $E$ has the
\emph{b-comparability property} if
\begin{enumerate}
\item[(a)] $E$ has  the b-property.
\item[(b)] For all $e,f\in E$ such that  $eCf$, the set
\[
P_\le(e,f):=\{p\in P(e,f),\ J_p(e)\le J_p(f) \text{ and } J_{p^\perp}(f)\le J_{p^\perp}(e)\}
\]
is nonempty.
\end{enumerate}
\end{definition}

The b-comparability property has important consequences on the set of projections and on the structure of the C-blocks.

\begin{theorem}\label{th:proper}{\rm \cite[Theorem 6.1]{Pucompr}} Let $E$  have the b-comparability property. Then
every sharp element is a projection: $P=E_S$.
\end{theorem}

\begin{theorem}\label{th:bcompar} { \rm  \cite[Theorem 7.1]{Pucompr}} Let $E$ have the b-comparability property and
 let $C=C(B)$ for a block $B$ of $P$. Then
\begin{enumerate}
\item $C$ is an MV-effect algebra.
\item For $p\in B$,  the restriction ${J}_p|_{C}$ coincides with $U_p$ (recall Example
\ref{ex:compr_mv}) and
$(U_p)_{p\in B}$ is the maximal compression base in $C$. Moreover, $(U_p)_{p\in B}$
 has the b-comparability property in $C$.
\item If $E$ has the projection cover property, then $C$ has the projection cover property.
\item If $E$ is $\sigma$-orthocomplete, then $C$ is $\sigma$-orthocomplete.
\end{enumerate}
\end{theorem}

Finally, we have the following definition of spectrality on effect algebras.

\begin{definition} An effect algebra $E$ with a given compression base $(J_p)_{p\in P}$ is \emph{spectral} if it has both the projection cover and the b-comparability property.
\end{definition}

It was proved in \cite[Thm. 4.15]{JP2} that in  an archimedean spectral effect algebra, any
element $a$  has a
\emph{spectral resolution} that can be characterized as the unique family
$\{p_\lambda\}_\lambda$ of projections commuting with $a$ parametrized by $\lambda\in \mathbb Q\cap [0,1]$,
 which is nondecreasing and right continuous (that is, $p_\lambda\le p_\mu$ if $\lambda\le
 \mu$ and
 $\bigwedge_{\lambda<\mu}p_\mu=p_\lambda$) and satisfies an additional condition that can
 be interpreted as ``$a\le \lambda$ on $p_\lambda$ and $a\ge \lambda$ on
 $p_\lambda^\perp$''. In addition, if the effect algebra has a separating set of states,
 then any element is uniquely determined by its spectral resolution and two elements
 commute if and only if the corresponding spectral resolutions are elementwise compatible.

\begin{example} \cite{JP2}
\begin{enumerate}
\item The algebra $E(\mathcal H)$ of Hilbert space effects is spectral,
similarly, the unit interval in a von Neumann algebra or in a JBW-algebra is spectral
\cite{AS}. By the results of \cite{JP1,JP2}, the unit interval in a JB-algebra is spectral
if and only if the JB-algebra is \emph{Rickart}.
\item  If $E$ and $F$ are spectral effect algebras, then their direct product $E\times F$
endowed with the direct product of compression bases is spectral.
\item Using a faithful state of $E(\mathcal H)$, the horizontal sum $E(\mathcal H)\dot{\cup} E(\mathcal H)$ can be endowed with a
compression base which makes it spectral. In general  the horizontal sum of spectral effect algebras is not spectral.
\item An MV-effect algebra is spectral if it is monotone $\sigma$-complete or a boolean
algebra.
\item An OMP is spectral if and only if it is a boolean algebra.

\end{enumerate}

\end{example}

\section{Special types of effect algebras}\label{sec:coseq}

\subsection{Convex effect algebras}

An effect algebra  $E$ is \emph{convex}  \cite{GuPu} if for every $a\in E$ and $\lambda \in [0,1] \subset {\mathbb R}$ there is an element $\lambda a\in E$ such that for all $a,b\in E$ and all $\lambda, \mu \in [0,1]$ we have

\begin{enumerate}
\item[(C1)] $\mu(\lambda a)=(\lambda \mu)a$.
\item[(C2)] If $\lambda + \mu\leq 1$ then $\lambda a\oplus \mu a\in E$ and $(\lambda +\mu)a=\lambda a\oplus \mu a$.
\item[(C3)] If $a\oplus b\in E$ then $\lambda a\oplus \lambda b \in E$ and $\lambda(a\oplus b)=\lambda a\oplus \lambda b$.
\item[(C4)] $1a=a$.
\end{enumerate}

A convex effect algebra is convex in the usual sense: for any $a,b\in E$, $\lambda \in [0,1]$, the element $\lambda a\oplus
(1-\lambda)b\in E$. An important example of a convex effect algebra is the algebra $E(\mathcal H)$ of Hilbert space
effects, Example \ref{ex:effects}.

Let $V$ be  an ordered real linear space with positive cone $V^+$. Let $u\in V^+$ and let us form the interval effect algebra $V[0,u]$.
 A straightforward
verification shows that  $(\lambda,x)\mapsto \lambda x$ is a convex structure on $V[0,u]$, so $V[0,u]$ is a
convex effect algebra which we call a \emph{linear effect algebra}. By \cite[Theorem
3.4]{GPBB}, any convex effect
algebra is isomorphic to the linear effect algebra $V[0,u]$ in an ordered vector space
with order unit $u$. Moreover, this isomorphism is
\emph{affine}, which means that it preserves the convex structures.

In convex effect algebras we have a stronger notion of archimedeanity:

\begin{definition} \label{de:strongarch} A convex effect algebra $E$ is \emph{strongly
archimedean} if, for any $a,b,c\in E$, if $a\leq b\oplus \frac{1}{n} c$ $\forall n\in
{\mathbb N}$, then  $a\leq b$.
\end{definition}

The next theorem describes the relations among order unit spaces, ordering sets of states and strongly archimedean convex effect algebras.

\begin{theorem}\label{th:archim}{\rm \cite[Theorem 3.6]{GPBB}} Let $E\simeq V[0,u]$ for an
ordered vector space $(V,V^+)$
with order unit $u$.  Then the following statements are equivalent. {\rm(a)} $E$ possesses an ordering set of states. {\rm(b)} $E$ is strongly archimedean. {\rm(c)} $(V, V^+,u)$ is an order unit space.
\end{theorem}

Recall that an order unit space $(V,V^+,u)$ is endowed with an order unit norm,  defined as
\[
\|v\|:=\inf\{\lambda>0:\ -\lambda u\le  v \le \lambda u\}.
\]
Note that the unit
interval $V[0,u]$ is norm-closed.  We will consider below the norm in $E\simeq V[0,u]$ inherited from $V$.

Let $E$ be a strongly archimedean convex effect algebra with the corresponding order unit
space $(V,V^+,u)$. Spectrality in order unit spaces in the sense of Foulis was studied in
\cite{FPspectres}. Let us recall that a compression on $(V,V^+,u)$ is defined as a positive
linear map $V\to V$ such that its  restriction is a compression  on the effect algebra
$E$. Similarly, a compression base is a collection of linear maps $(J_p)_{p\in P}$ such that their
restrictions form a compression base  in $E$. For $p\in P$ and $v\in V$, we define the
commutants
\[
C(p)=\{v\in V:\ v=J_p(v)+J_{p^\perp}(v)\},\qquad PC(v)=\{p\in P:\ v\in C(p)\}
\]
and the  bicommutant $P(v)=PC(PC(v)\cup \{v\})$. We then say that $(V,V^+,u)$
has the comparability property if
the set
\begin{equation}\label{eq:Vcomparability}
P_\pm(v):=\{p\in P(v):\ J_{p^\perp}(v)\le 0\le J_p(v)\}
\end{equation}
is nonempty.

\begin{definition}\label{def:orthogonal} An orthogonal decomposition of  $v\in V$ is a
decomposition of the form $v=v_+-v_-$, where $v_+,v_-\in V^+$ and there is a projection
$p\in PC(v)$ such that $v_+=J_p(v)$, $v_-=-J_{p^\perp}(v)$.

\end{definition}

By \cite[Thm. 3.2 and Lemma 4.2]{Fgc}, if $(V,V^+,u)$ has the comparability property, then
each element has a unique orthogonal decomposition, determined by any projection in
$P_{\pm}(v)$.

We say that $(V,V^+,u)$ has the
projection cover property if  $E\simeq V[0,u]$ has the projection cover property, with
the restricted compression base. If the comparability property holds, then the projection
cover property is equivalent to existence of a \emph{Rickart mapping} \cite[Theorem 2.1]{FPspectres}, defined as a map ${\, }^*: V\to P$,
 where $v^*$ is the (necessarily unique) projection such that
 \[
p\in P,\ p\le v^* \iff v\in C(p),\text{ and } J_p(v)=0.
\]
For $a\in E$, we have $a^*=(a^\circ)^\perp$ and $(\lambda a)^\circ=a^\circ$ for any
$\lambda\in [0,1]$. More generally, let $v\in V^+$, then we may define the \emph{support} of $v$ as
$v^\circ:=(v^*)^\perp$ and it is easily seen that $v^\circ=(cv)^\circ$ for any $c\in
\mathbb R^+$.

We say that the order unit space $(V,V^+,u)$ is spectral if it has both the projection cover and the
comparability property. In this case, every element $v\in V$ has a unique spectral
resolution $\{p_{v,\lambda}\}_{\lambda\in {\mathbb R}}\subseteq P(v)$, where the spectral
projections are defined as \cite{FPspectres}
\begin{equation}\label{eq:spectprojs}
p_{v,\lambda}:=((v-\lambda)_+)^*.
\end{equation}
The spectral resolution of $v$ is continuous from the right in the sense that 
if $\alpha \in {\mathbb R}$, then $p_{v,\alpha} = \bigwedge\{p_{v,\lambda}: \alpha <
\lambda\in {\mathbb R}\}$, \cite[Theorem 3.5]{FPspectres}.

An element $\lambda\in \mathbb R$ is called an    \emph{eigenvalue} of $v$ if the
projection  $d_{v,\lambda}:= (v-\lambda)^*$ is nonzero, in this case, $d_{v,\lambda}$ is called
the \emph{$\lambda$-eigenprojection} of $v$.  For $\alpha\in \mathbb R$, the projection
$d_{v,\alpha}$ may be interpreted as the "jump" that occurs in the
spectral resolution as $\lambda$ approaches $\alpha$ from the left in the following sense:
 $p_{v,\alpha}-d_{v,\alpha} = \bigvee \{p_{v,\lambda}: \alpha >\lambda \in {\mathbb R}\}$, \cite[Theorem 3.6]{FPspectres}.

The \emph{spectral lower and upper bounds for $v$} are defined by $L_v:=\sup\{ \lambda \in {\mathbb R}:\lambda u\leq v\}$ and $U_v:=\inf\{ \lambda \in {\mathbb R}: v\leq \lambda u\}$, respectively. By \cite[Theorem 3.3 (vii)]{FPspectres}, $L_v=\sup\{\lambda \in {\mathbb R}: p_{v,\lambda} = 0\}$ and $U_v=\inf\{ \lambda \in {\mathbb R}: p_{v,\lambda}=u\}$.
 Then $v$ can be written as a Riemann-Stieltjes type integral
\begin{equation}\label{eq:spectresV}
v=\int_{L_{v}-0}^{U_v} \lambda dp_{v,\lambda}.
\end{equation}
For more details about spectral resolutions  see \cite{FPspectres}.

\begin{example}
Let $V$ be the space of self-adjoint operators on a Hilbert space $\mathcal H$ and let $V^+$ be
the cone of positive operators. Then $(V,V^+,u=I)$ is a spectral order unit space. In this
case, the Rickart mapping sends each $v\in V$ onto its kernel projection. Thus, we obtain
the usual definition of the spectral resolutions, eigenvalues and eigenprojections.

\end{example}

It was proved in \cite[Thm. 5.11]{JP2} that a strongly archimedean convex effect algebra $E$ has the
comparability property or projection cover property if and only if the corresponding order
unit space $(V,V^+,u)$ has the same
property. In particular, $E$ is spectral if and only if $(V,V^+,u)$ is spectral. In this
case, every element  $a\in E$ has an integral representation of the form
\eqref{eq:spectresV} with respect to a unique spectral resolution
$(p_{a,\lambda})_{\lambda \in [0,1]}\subseteq P(a)$, and this spectral resolution is the
same as obtained in $(V,V^+,u)$. Any element is uniquely  determined by its spectral resolution and two
elements in $E$ commute if and only if the corresponding spectral resolutions are
elementwise compatible. We will also need the following two results.

\begin{corollary}\label{coro:limit} Let $E$ be a strongly archimedean convex and spectral
effect algebra. Then every element $a\in E$ is the norm limit and supremum of an ascending
sequence $a_n\le a_{n+1}$ of elements of the form
\[
a_n=\oplus_i c_{n,i} p_{n,i},
\]
with $c_{n,i}\in [0,1]$ and $p_{n,i}\in P(a)$, $\oplus_i p_{n,i}=1$.

\end{corollary}

\begin{proof} Let $(V,V^+,u)$ be the spectral order unit space such that $E\simeq V[0,u]$
(with the unique extended compression base). By \cite[Cor. 3.1]{FPspectres}, any element
 $a\in V[0,u]$ is the norm limit of an ascending sequence of elements of the
 form $a_n=\sum_i c_{n,i}p_{n,i}$, with $p_{n,i}\in P(a)$ and $c_{n,i}\in \mathbb R$, and
 by \cite[Lemma 5.1]{FPspectres} we may assume that $\sum_ip_{n,i}=u$.
 Since $a_n\le a_{n+1}$, we have $a_n\le a\le u$. We now may replace each $a_n$ by an
 element $(a_n)_+$, which is obtained by putting negative coefficients $c_{n,i}$ to zero.
 All the projections $p_{n,i}$ are in $P(a)$, and therefore mutually commuting and
 commuting also with $a$ and all $a_n$. Put $p_n=\oplus_{i, c_{n,i}>0} p_{n,i}$, then
 \[
(a_n)_+=p_na_n\le p_n a\le a.
 \]
Similarly, from $a_n\le a_{n+1}\le (a_{n+1})_+$, we obtain that $(a_n)_+\le (a_{n+1})_+$. Notice also that
\[
0\le a-(a_n)_+\le a-a_n\to 0.
\]
We have obtained an ascending sequence  $\{(a_n)_+\}$ of elements in $[0,u]$ that converges
 to $a$ in norm. It is clear that $(a_n)_+$ is again a simple element and the remaining
 coefficients $c_{n,i}$ must be in the interval $[0,1]$ (this can be observed
 e.g. from the fact that for every sharp element $p$ of $E$ there is some state $s$ on $E$
 such that $s(p)=1$, \cite{GPBB}). To conclude the proof, note that since the positive
 cone $V^+$ is norm-closed, the norm limit of an ascending sequence is its supremum.

\end{proof}

\begin{lemma}\label{lemma:covex_floor} Let $E$ be spectral and let ${\,}^*$ be the Rickart mapping in the
corresponding spectral order unit space $(V,V^+,u)$. Let $(p_\lambda)_{\lambda\in [0,1]}$ be the
spectral resolution of $a\in E$. Then
\begin{align*}
1&=p_1,\quad (a^\circ)^\perp=p_0,\\
a_\circ&=(a-1)^*=\bigwedge_{\lambda<1}p_\lambda^\perp.
\end{align*}

\end{lemma}

\begin{proof} The equality for $p_0$ and $p_1$ follow easily from the definition.
Let $d:=d_1=(a-1)^*$, then  $d\in P$ and $J_d(a-1)=0$, so that
$J_d(a)=d$ and hence $d\le a$. If $q\in P$ is any projection such that $q\le a$, then $q$
commutes with $a$ and $J_q(a-1)=0$, so that $q\le d$ by definition of the Rickart mapping.
It follows that $d=a_\circ$. For the second equality, we have by \cite[Thm. 3.6]{FPspectres}
that $\bigvee_{\lambda<1} p_\lambda$ exists and equals $p_1-d=d^\perp$. This implies that
\[
d=(\bigvee_{\lambda<1}p_\lambda)^\perp= \bigwedge_{\lambda<1} p_\lambda^\perp.
\]

\end{proof}

\subsection{Sequential effect algebras}

\begin{definition}\label{def:SEA} A \emph{sequential effect algebra} (SEA) \cite{GuGr} $(E:+,1,0,\circ)$ is an effect algebra with an additional \emph{sequential product} operation $\circ$. We denote $a|b$ when $a\circ b=b\circ a$ (i.e. when $a$ and $b$ commute). The sequential product is required to satisfy the following axioms.
\begin{enumerate}
\item[(S1)] $a\circ(b+c)=a\circ b+ a\circ c$.
\item[(S2)] $1\circ a=a$.
\item[(S3)] $a\circ b =0 \ \implies b\circ a=0$.
\item[(S4)] If $a|b$ then $a|b^\perp$ and $a\circ (b\circ c)=(a\circ b)\circ c$ for all $c$.
\item[(S5)] If $c|a$ and $c|b$ then $c|(a\circ b)$ and if $a+b$ is defined $c|(a+b)$.
\end{enumerate}
\end{definition}

\begin{definition}\label{de:sigmasea} {\rm \cite{GuGr}}. A $\sigma$-SEA is a SEA which is
monotone $\sigma$-complete (hence a $\sigma$-effect algebra)
 such that if $a_1\geq a_2\geq \cdots $ then $b\circ\wedge a_i=\wedge(b\circ a_i)$ and if $b|a_i$ for all $i$ then $b|\wedge a_i$.
\end{definition}

\begin{lemma}\label{le:sharp} \cite{GuGr} Let $p,a\in E$ with $p$ sharp.
\begin{enumerate}
\item $a$ is sharp iff $a\circ a^\perp=0$ iff $a\circ a=a$.
\item $p\leq a$ iff $p\circ a=a\circ p =p$.

\item $a\leq p$ iff $p\circ a=a\circ p=a$

\item $p\circ a=0$ iff $p+a$ is defined and in this case $p+a$ is the lest upper bound of $p$ and $a$. The sum $p+a$ is sharp iff $a$ is sharp.
\item $a|p$ iff $a\leftrightarrow p$.

\item If $a|p$ then $p\circ a=p\wedge a$.
	
\end{enumerate}
\end{lemma}

It is easy to check that any map on SEA of the form $J_p(a)=p\circ a$, $p\in E_S$, is a
compression with focus $p$ and a supplement $J_{p^\perp}(a)=p^\perp\circ a$.
Moreover, by \cite[Theorem 3.4]{Gudcomprba}, $(J_p)_{p\in P}$, $P=E_S$,  is a maximal compression base for $E$.
Below we will always assume that $E$ is endowed with this compression base.

Let $S\subseteq E$ be a subset of elements of $E$, then $S':=\{ a\in E: s|a,\ \forall s\in S  \}$ is the \emph{commutant} of $S$. Similarly the \emph{bicommutant} $S'':=(S')'$ of $S$ is the set of all elements in $E$ that commute with every element in $S'$.

\begin{lemma}\label{le:commutants} Let $E$ be a SEA, $S\subseteq E$.
\begin{enumerate}
\item $S'$ is a sub-SEA of $E$.
\item If $p\in P$, then $C(p)=\{p\}'$.
\item If $S$ is a set of mutually commuting elements, then $S''$ is a commutative sub-SEA
of $E$.

\end{enumerate}

\end{lemma}

\begin{proof}
By \cite[Lemma 3.1]{GuGr}, $0,1\in S'$. The rest of (i) follows immediately from the axioms
(S4), (S5)
of SEA, Definition \ref{def:SEA}. By Lemma \ref{le:sharp}, for any $p\in P$ and $a\in E$,
$a\in \{p\}'$ if and only if $a\leftrightarrow p$. Statement (ii) now follows by Lemma
\ref{le:comE}. For (iii), using (i), it is enough to prove that $S''$ is commutative.
Since $S$ is commutative, $S\subseteq S'$, which yields $S''\subseteq  S'=S'''$, which by definition yields commutativity of $S''$.

\end{proof}

Notice that the maximal compression base in  $S'$ coincides with $(J_p|_{S'})_{p\in P\cap
S'}$.

\section{Spectrality in convex sequential effect algebras}\label{sec:main}

Throughout  this section, we consider a SEA $E$ which is also a strongly archimedean convex effect
algebra. Note that in this case, the sequential product is affine in the second variable,
that is, for $a,b,c\in E$ and $\lambda\in [0,1]$, we have
\[
a\circ(\lambda b\oplus (1-\lambda)c)=\lambda a\circ b\oplus (1-\lambda)(a\circ c),
\]
this follows  from axiom (S1) and \cite[Thm. 5.8]{JP2}. By Theorem
\ref{th:archim}, we see that $E$ is isomorphic to the unit interval in an order unit space
$(V,V^+,u)$, moreover, for any $a\in E$, the map $b\mapsto a\circ b$ extends to a positive
linear map on $(V,V^+,u)$. This also implies that  the sequential product is
continuous in the inherited order unit norm in the second variable.
If $E$ is commutative, then $(a,b)\mapsto a\circ
b$ extends to a positive bilinear product on $(V,V^+,u)$

\begin{example}\label{ex:hilb_seq} It is easily seen that the Hilbert space effect algebra
$E(\mathcal H)$ is convex, strongly archimedean and sequential, with the sequential
product given by
\[
a\circ b=a^{1/2}ba^{1/2}.
\]
In fact, $E(\mathcal H)$ is a prototypical example of both sequential and convex effect
algebra. Moreover, $E(\mathcal H)$ is monotone complete and spectral, and we have $a|b\iff
aCb \iff ab=ba$ for any $a,b\in E(\mathcal H)$. Note that the
sequential product  in $E(\mathcal H)$ is not unique \cite{WeJu}, but  all of
these products must coincide on pairs $(p,a)$ where $p\in P(\mathcal H)$.

\end{example}

\begin{example}\label{ex:jb} More generally, let $A$ be a JB-algebra with unit $u$ and Jordan product
$(a,b)\mapsto a*b$, \cite[Chap. 1]{AlSh}, \cite{HO}.
Let $E$ be the unit interval $[0,u]$ in $A$. For $a,b\in E$, we define
\[
a\circ b:=2a^{1/2}*(a^{1/2}*b)-a*b.
\]
It was proved in \cite{Wetcom} that $E$ with this product is a (convex,
strongly archimedean) SEA.

\end{example}

Our aim is to study spectrality for this type of effect algebras. Let us first look at the
case when $E$ is commutative. We will need the following representation result. Below,
$C(X)=C(X,\mathbb R)$ denotes the space of continuous functions $X\to \mathbb R$ and
$C(X,[0,1])$ denotes the set of continuous functions $X\to [0,1]$.

\begin{theorem}[Kadison]\cite{kadison} \label{th:kad}  Let $V$ be an order unit space with a bilinear
commutative operation $\circ$ such that $1\circ v=v$ and $v\circ w\geq 0$ whenever $v,w\geq 0$, then there exists a compact Hausdorff space $X$ and an isometric embedding $\Phi:V\to C(X)$ such that $\Phi(X)$ lies dense in $C(X)$ and $\Phi(v\circ w)=\Phi(v)\Phi(w)$. If $V$ is complete in its norm, then $V$ is isomorphic as an ordered algebra to $C(X)$.
\end{theorem}

\begin{corollary}\label{coro:commutative} Let $E$ be a commutative SEA which is also convex and strongly archimedean.
Then $E$ is isomorphic to a dense subalgebra of the algebra $C(X,[0,1])$ of continuous
functions $X\to [0,1]$ for some compact Hausdorff space $X$. Moreover, the following are
equivalent.
\begin{enumerate}
\item  $E$ is monotone $\sigma$-complete;
\item $E\simeq C(X,[0,1])$, with $X$ basically disconnected;
\item $E$ is norm-complete and spectral.

\end{enumerate}

\end{corollary}

\begin{proof} The first statement follows by the Kadison theorem and the remarks above.
The rest follows by \cite[Example 5.13]{JP2}.

\end{proof}

Let us now turn to the general case.  One of the problems that
appear in this setting is that if the b-property holds, we have two notions of commutativity, namely $aCb
$ and $a|b$. Note that for $p\in P$, $a|p$ $\iff$  $aCp$ $\iff$ $a\leftrightarrow p$, but
for general elements these two notions might  be distinct and therefore  should not be confused.  We will show
in the course of our characterization of spectrality that these two notions are equal
under some additional conditions. Note that this is true for the algebra of Hilbert space
effects,  Example
\ref{ex:hilb_seq}.

We start by observing the following property of the projection cover, which will be needed in the
sequel.

\begin{lemma}\label{lemma:projcover} Let $a,b\in E$ be such  that $a^\circ$ exists and
$b=\lim_n b_n$, where for each $n$, $b_n\le b$ and $b_n=\oplus_i
\lambda_{n,i}p_{n,i}$ with
$\lambda_{n,i}\in [0,1]$ and $p_{n,i}\in P$. Then $a\circ b=0$ if and only if $a^\circ \circ
b=0$.

\end{lemma}

\begin{proof} Assume  $a^\circ \circ b=0$, then by axiom (S3), $b\circ a^\circ=0$, so that
$a\circ b=0$ using axioms (S1) and (S3) together with the fact that $a\le a^\circ$. For
the converse, assume first that $b\in P$. By Lemma \ref{le:sharp} (iv), $a\circ b=0$ implies
$a\le b^\perp$, so that $a^\circ \le b^\perp$ and hence $a^\circ \circ b=0$. Next, let $b=\oplus_i \lambda_{i}p_i$
 for some $\lambda_i\in [0,1]$ and projections $p_i$, then $a\circ b=0$ implies that
 $\lambda_i(a\circ p_i)=a\circ \lambda_ip_i=0$, so that $a\circ p_i=0$ for all $i$ such
 that $\lambda_i>0$. By the previous step, $a^\circ \circ p_i=0$  and hence
 $a^\circ \circ b=0$. Finally, let $b=\lim_n b_n$ as in the assumption, then $a\circ
 b_n\le a\circ b=0$ implies that $a\circ
 b_n=0$  and hence $a^\circ \circ b_n=0$ for all $n$. The proof follows by continuity of
 the sequential product in the second variable.

\end{proof}

In addition to our standing assumptions in this section, we will always assume the
following property which will be
called \emph{property A}: For every ascending sequence  $a_n\le a_{n+1}$
of mutually commuting elements in $E$ such that  $\vee_na_n$ exists and $b\in E$,
$a_n|b$ for all $n$ implies that $\vee_na_n|b$.


We next observe some consequences of property A. The first result shows that this property  ensures  that the sequential product
is in agreement  with the convex structure.

\begin{lemma}\label{le:aff} Let $a,b\in E$, $\lambda\in [0,1]$.
Then
\begin{enumerate}
\item  $a\circ (\lambda b)=(\lambda
a)\circ b=\lambda(a\circ b)$.
\item  If $a|b$ then $a|\lambda b$.

\end{enumerate}

\end{lemma}

\begin{proof}
(i) Since $\circ$ is affine in the second variable, we have $a\circ(\lambda b)=\lambda(a\circ b)$.
The rest of the proof of (i) uses similar arguments as the proof of \cite[Proposition
3.9]{Wet}. First, note that $\frac1n a|\frac1n a$, so that $\frac1n a|a$, by axiom (S5).
Similarly we get  $\gamma a|a$ and also $\gamma a^\perp|a^\perp$ for all rationals
$\gamma\in [0,1]$. By axioms (S4)
and (S5) then $\gamma a^\perp|a$ and $a|(\gamma a\oplus \gamma a^\perp)$, so that
$a|\gamma 1$. It follows that $(\gamma 1)\circ
a=a\circ (\gamma 1)=\gamma (a\circ 1)$, by the first part of the proof. Let now
$\gamma_i\in [0,1]$ be an increasing
sequence of rationals such that $\vee_i\gamma_i=\lambda$, then $\vee_i \gamma_i1=\lambda
1$ and by property A,  we obtain that $a|\lambda1$. We now
compute using (S4)
\[
(\lambda a)\circ b=(a\circ (\lambda 1))\circ b=a\circ ((\lambda 1)\circ b)=a\circ (\lambda
b)=\lambda (a\circ b).
\]
The statement (ii) is immediate from (i).

\end{proof}

\begin{lemma}\label{lemma:closed} Let $S\subseteq E$ be any subset. Then $S'$ is closed
under norm limits of sequences of mutually commuting elements.

\end{lemma}

\begin{proof} We will use an argument inspired by \cite{Fou}. So
let $a_n\in S'$ be any
norm-convergent sequence of mutually commuting elements and let $a=\lim_na_n$. By restriction to a subsequence, we may assume that
$\|a_{n+1}-a_n\|< 2^{-n}$ for all $n$. Put $s_n:=\frac12(a_n+(1-2^{-n})u)$, then $s_n$ is
an ascending sequence of mutually commuting elements in $S'$. Indeed, it is enough to note that if we put
$b_k=a_{k+1}-a_k+2^{-k}u$, $k=1,2,\dots$, the assumption on
$\{a_n\}$ implies that $0\le b_k\le 2^{1-k}u\le u$, and $s_n=\frac12(a_1+
\sum_{k=1}^nb_k)$. Moreover, $\lim_n s_n=\frac12(a+u)$. Since the norm-limit of an
ascending sequence is its supremum,  property A  implies that $\frac12(a+u)\in S'$. Using
\cite[Lemma 3.1 (v)]{GuGr}, we obtain that $\frac12 a\in S'$ and consequently also $a\in
S'$.

\end{proof}

\begin{lemma}\label{le:commut}   Let $S\subseteq E$ be a subset of mutually commuting
elements.
Then $S''$ is a norm-closed strongly archimedean convex commutative sub-SEA of $E$.
\end{lemma}

\begin{proof} By Lemma \ref{le:commutants} $S''$ is a commutative sub-SEA of
$E$.  By Lemma \ref{le:aff}, $S''$ is also convex, with the convex structure inherited
from $E$, and it is easily seen that it must be strongly archimedean.
The fact  that $S''$ is norm-closed follows by Lemma \ref{lemma:closed}.

\end{proof}

\begin{prop}\label{prop:comp_com} Assume that $E$ has the b-comparability property. Then
$aCb$ implies $a|b$.

\end{prop}

\begin{proof} We have  $aCb$ $\iff$ $P(a)\leftrightarrow b$ $\iff$  $P(a)\in \{b\}'$.
By \cite[Theorem 3.22]{JP1}, $a$ is in the closed linear span of
$P(a)$, here we potentially have to consider the extension to $(V,V^+,u)$. Let $a_n\in
\mathrm{span}(P(a))$, $a_n\to a$. By replacing $a_n$ by $\|a\|\frac{a_n}{\|a_n\|}$, we may
assume that $-u\le a_n\le u$, so that $c_n:=\frac12(a_n+u)$ is a sequence in
$\mathrm{span}(P(a))\cap E$ converging to $\frac12(a+u)$. It follows that any $c_n$ is of
the form
$c_n=\oplus_i\lambda_{n,i}p_{n,i}$ for $\lambda_{n,i}\in [0,1]$, $p_{n,i}\in P(a)\subseteq
\{b\}'$, so that  $c_n$ is a sequence of mutually commuting elements in $\{b\}'$.  Lemma
\ref{le:commut} now implies that $\lim_nc_n=\frac12(a+u)\in \{b\}'$, hence also $a|b$.

\end{proof}

We now want to look at the opposite implication of Proposition \ref{prop:comp_com}. We
will show, after some preparations, that it holds if $E$ is spectral and norm-complete (see Proposition \ref{prop:commut} below).

\begin{lemma}\label{lemma:floor} Assume that $E$ is spectral and let $a^k=a\circ\dots
\circ a$. Then $\{a^k\}$ is a descending sequence of commuting elements in $E$ and $a_\circ=\bigwedge_k a^k$.

\end{lemma}

\begin{proof} It is clear that $\{a^k\}$ is a descending sequence of commuting elements in
$E$. By definition and Lemma \ref{le:sharp} (ii), $a_\circ \circ a=a_\circ$, so  that $a_\circ\circ a^k=a_\circ$ and
this shows that $a_\circ\le a^k$ for all $k\in \mathbb N$. Let $b\in E$ be any element
such that $b\le a^k$ for all $k$. Then $a\circ b\le a\circ a^k=a^{k+1}$ for all $k$.
Let $\{p_\lambda\}_{\lambda\in [0,1]}$ be the spectral resolution for $a$
and put $a_\lambda=p_\lambda\circ a$. Then $a_\lambda\le \lambda p_\lambda$  \cite[Theorem 3.3 (ii)]{FPspectres} and
consequently for all $k\in \mathbb N$,
\[
a_\lambda\circ b=p_\lambda\circ(a\circ b)\le p_\lambda\circ a^{k+1}=(p_\lambda\circ
a)^{k+1}\le \lambda^{k+1}p_\lambda,
\]
here we used axiom (S3) and the fact that $p_\lambda\in P$ and $a$ commute. By
archimedeanity, for $\lambda<1$ this implies that $a_\lambda\circ b=0$. Similarly as
before, we obtain
\[
0=a_\lambda\circ b=(a\circ p_\lambda)\circ b=a\circ (p_\lambda\circ b).
\]
Since $E$ is spectral, we see from Corollary \ref{coro:limit} that the assumptions in
Lemma \ref{lemma:projcover} are satisfied, so that we obtain
\[
0=a^\circ \circ (p_\lambda\circ b)=(a^\circ \circ p_\lambda)\circ b=(p_\lambda\circ
a^\circ)\circ
b=p_\lambda\circ (a^\circ\circ b)=p_\lambda\circ b,
\]
the last equality follows from the fact that $b\le a\le a^\circ$.
This implies that $b\le p_\lambda^\perp$ for all $\lambda<1$ and hence
\[
b\le \bigwedge_{\lambda<1} p_\lambda^\perp=a_\circ,
\]
by Lemma \ref{lemma:covex_floor}.
\end{proof}

The next result shows that
under some further assumption on the bicommutants in $E$, any element in the order unit space
$(V,V^+,u)$ has a suitable decomposition into positive and negative part.  Note that since
we have $-\|v\| u\le v \le \|v\| u$ for any $v\in V$, it follows that $(2\|v\|)^{-1}(v+\|v\| u)\in V[0,1]\simeq E$.

\begin{prop}\label{prop:decomp}  Let $v\in V$ and let $b=(2\|v\|)^{-1}(v+\|v\| u)$. Assume that  the bicommutant $\{b\}''$
is norm-complete. Then there are some
$\mu_\pm>0$ and elements $a_\pm \in \{b\}''$ such that $a_+\circ a_-=0$ and
\[
v=\mu_+a_+-\mu_-a_-.
\]

\end{prop}

\begin{proof} By Lemma  \ref{le:commut}, $\{b\}''$ is a strongly archimedean convex
commutative sub-SEA of $E$.  Since it is also
norm-complete by the assumption, we have by Corollary
\ref{coro:commutative} that  $\{b\}''\simeq C(X,[0,1])$ for some compact Hausdorff space $X$.
Since $C(X)$ is spanned by $C(X,[0,1])$, it corresponds to the subspace in $V$ spanned by
$\{b\}''$.

Since $v$ is a linear combination of $b$ and $u$, there is a corresponding function $f\in
C(X)$. Put
$f_\pm=\frac12(|f|\pm f)$, then $f=f_+-f_-$, $f_\pm$ are positive elements in $C(X)$ and
we have $f_+f_-=0$.  It follows that there are some $g_\pm\in C(X,[0,1])$ and $\mu_\pm>0$
such that $f_\pm=\mu_\pm g_\pm$ and consequently $g_+g_-=0$. Let now  $a_\pm\in \{b\}''$ be the
elements corresponding to $g_\pm$. Then we have $v=\mu_+a_+-\mu_-a_-$ and
$a_+\circ a_-$follows from the fact that the sequential product in
$\{b\}''$ corresponds to the pointwise product of functions in $C(X,[0,1])$.

\end{proof}

\begin{lemma}\label{lemma:projections} Let $E$ be spectral, $a\in E$
and let $\{p_\lambda\}_{\lambda\in [0,1]}$ be the spectral resolution of $a$. Assume that
$\{a\}''$ is norm-complete. Then
$p_\lambda\in \{a\}''$, $\lambda\in [0,1]$.

\end{lemma}

\begin{proof} Let $b\in \{a\}''$, then $b^k\in \{a\}''$ for all $k\in \mathbb N$ and we see by Lemma \ref{lemma:floor}  and property A
that $b_\circ\in \{a\}''$. Since $b^\circ=(b^\perp)_\circ^\perp$, we have  $b^\circ\in
\{a\}''$. For $\lambda\in [0,1]$ put $v_\lambda:=a-\lambda u$, then we see that
$\{(2\|v_\lambda\|)^{-1}(v_\lambda+\|v_\lambda\| u)\}''=\{a\}''$ is norm-complete, so that
 we may apply Proposition \ref{prop:decomp}.  We  obtain a decomposition
\begin{equation}\label{eq:ortho}
a-\lambda u=\mu_+c_+-\mu_-c_-
\end{equation}
with $\mu_\pm>0$ and $c_\pm\in \{a\}''$, $c_+\circ c_-=0$.  We next show that this is
 an orthogonal   decomposition  of $a-\lambda u$ in $(V,V^+,u)$, in the sense of
Definition \ref{def:orthogonal}. Since such a decomposition is unique, this will imply that
$\mu_\pm c_\pm=(a-\lambda u)_\pm$ and hence
\[
p_\lambda=(a-\lambda u)_+^*=((\mu_+c_+)^\circ)^\perp=(c_+^\circ)^\perp\in\{a\}''.
\]

So let us choose $q=c_+^\circ$. Then $q\circ c_+=c_+$. Moreover, spectrality and Corollary
\ref{coro:limit} shows that the assumptions of Lemma \ref{lemma:projcover}  are satisfied.
It follows that
$q\circ c_-=0$, so that $c_-\le q^\perp$. We therefore have  $J_q(a-\lambda u)=\mu_+c_+$ and
$J_{q^\perp}(a-\lambda u)=-\mu_-c_-$, which shows that \eqref{eq:ortho} is indeed an
orthogonal decomposition.
This finished the proof.

\end{proof}

\begin{prop}\label{prop:commut} Let $E$ be norm-complete and spectral. Then for $a,b\in
E$, $a|b$ if and only if $aCb$.

\end{prop}

\begin{proof} Assume that $a|b$. Since $E$ is norm-complete and $\{a\}''$ is norm-closed
by Lemma \ref{le:commut},
we may apply Lemma \ref{lemma:projections}. It follows  that all spectral projections
$p_{a,\lambda}\in \{a\}''$,
so that $p_{a,\lambda}|b$ for all $\lambda$. Similarly, $p_{b,\mu}|p_{a,\lambda}$ for all
$\lambda,\mu$. Since $p|q$ is the same as $pCq$ for projections $p,q$,  this implies that $aCb$.
The converse follows from Proposition \ref{prop:comp_com}.

\end{proof}

We now prove our  main result.

\begin{theorem} \label{thm:sigmacomplete_spectral} Let $E$ be a strongly archimedean
convex SEA with property A.
  If every maximal commutative subalgebra is monotone $\sigma$-complete, then $E$ is spectral.
  If in addition $E$ is norm-complete, the converse also holds.

\end{theorem}

\begin{proof}  Assume that every maximal commutative subalgebra in $E$ is monotone
$\sigma$-complete. We will first show the projection cover property. For this, we prove that for every
$a\in E$ there exist a largest projection $p\in P$ such that $a\circ p=0$. It is then
clear that $p^\perp$ will be the projection cover of $a$.
The proof is similar to a proof in \cite{SW}. So let
\[
\mathcal P=\{e\in P,\ a\circ e=0\}.
\]
We will show that $\mathcal P$ is upward directed and that every increasing chain in
$\mathcal P$ has an upper bound. By the  Zorn lemma, $\mathcal P$ then has a maximal element,
which must be the largest elements since $\mathcal P$ is upward directed.

So let $p,q\in \mathcal P$ and put $\frac 12(p+q)=b$. Then $a\circ b=0=b\circ a$ by axiom
(S3), so that there is some maximal commutative subalgebra $M\subseteq E$ containing $a$ and $b$.
By the assumption and Corollary \ref{coro:commutative}, $M$ is norm-complete and spectral,
so that  $a$, $b$ and $M$ satisfy the conditions in Lemma \ref{lemma:projcover}, see Corollary
\ref{coro:limit}. Put $s$ be the projection cover of $b$ in $M$, then it follows  that
$a\circ s=0$, so that $s\in \mathcal P$. We
also have $\frac 12 p,\frac12 q\le b\le s$ and hence $p,q\le s$, as $s$ is a principal
element. It follows that $\mathcal P$ is upward directed.

Let now $\mathcal C$ be an increasing chain in $\mathcal P$, then all elements in
$\mathcal C$ mutually commute and also commute with $a$. Therefore there is a maximal
commutative subalgebra $M_1\subseteq E$ containing $a$ and $\mathcal C$. Let $s_1$ be the
projection cover of $a$ in $M_1$, then $s_1^\perp=1-s_1$ is the largest element in $P\cap M_1$ such
that $a\circ s_1^\perp=0$, so that $s_1^\perp\in \mathcal P$ and $\mathcal C\le
s^\perp_1$,
which means that $\mathcal C$ has an upper bound in $\mathcal P$. This proves that
$a^\circ$ exists.

For b-comparability, it will be  enough to show that the corresponding order unit space $(V,V^+,u)$ has
the comparability property. Recall that this means that for any $v\in V$, the set
$P_\pm(v)$ defined in \eqref{eq:Vcomparability} is nonempty.
So choose some $v\in V$. Note that for any $b\in E$, the bicommutant is contained in some
maximal commutative subalgebra, so that $\{b\}''$ is norm-complete by the assumption,
Corollary \ref{coro:commutative} and
Lemma \ref{le:commut}. Hence we may apply Proposition \ref{prop:decomp}. Let
$v=\mu_+c_+-\mu_-c_-$ be the obtained decomposition. By the previous paragraph, $E$ has
the projection cover property, so we may put $p=c_+^\circ$. Now observe that the
conditions of Lemma \ref{lemma:projcover} are satisfied for any $a,b\in E$. Indeed,
$a^\circ$ exists and $b$ is contained in some maximal commutative subalgebra $M$,
which is monotone $\sigma$-complete, hence spectral (Corollary
\ref{coro:commutative}), so that   $b$ can be obtained as the
norm limit of an ascending sequence of simple elements $b_n\in M$ (Corollary
\ref{coro:limit}). It follows that $p\circ c_+=c_+$ and $p\circ c_-=0$, so that
$p^\perp\circ c_-=c_-$. It is now easily checked that $p\in P_\pm(v)$.

To prove the converse, assume that $E$ is norm-complete and spectral. By Proposition \ref{prop:commut}, we see that
maximal commutative subalgebras are the same as C-blocks. The statement follows by
\cite[Thm. 3.33]{JP1}.

\end{proof}

The next result follows immediately from the above theorem and the definitions of a
$\sigma$-SEA (cf. \cite{wetering}).

\begin{corollary}\label{coro:sigma_sea} Any $\sigma$-SEA is spectral.

\end{corollary}

\subsection{Context-spectrality in convex SEAs}

Let us now turn to the notion of spectrality  in convex SEAs defined in terms of contexts
as in \cite{Gucs}, which we first briefly recall.
So let $E$ be a convex effect algebra. An element $a\in E$ is \emph{one dimensional} if
for $b\in E$, $b\le a$ implies that $b=ta$ for some $t\in [0,1]$. A  \emph{context} in $E$ is a finite  set of
one dimensional sharp elements
$p_1,\dots,p_n$ such that $\oplus_{i=1}^n p_i=1$. In \cite{Gucs}, $E$ is called spectral if
every element $a\in E$ has the form $a=\oplus_i \mu_i p_i$ for some context
$\{p_1,\dots,p_n\}$
and $\mu_i\in [0,1]$. We will call such effect algebra \emph{context-spectral}, to
distinguish this notion from  spectrality considered in the present paper.
We will show that if $E$ is also sequential, then context-spectrality is stronger than
spectrality.

We will, in fact, use a weaker assumption that every element $a\in E$ is \emph{simple}, that
is, it can be written as a sum $a=\oplus_i \mu_i p_i$ for some sharp elements $p_1,\dots,
p_n$, $\oplus_i p_i=1$ as before, but  $p_i$ are not assumed to be one-dimensional.
(Note that, as in the case of context-spectrality, the number of the sharp elements $p_i$
in such decompositions  might be different). In this case, there is always an expression
of this form with $\mu_1<\dots <\mu_n$, in \cite{FPspectres}, such an
expression is called a reduced representation of the simple element.  For general convex
effect algebras, it is not clear whether simple elements have a unique reduced
representations. Such uniqueness was proved under an additional assumption in
\cite[Proposition 10]{JePl}.

Let $E$ be a convex effect algebra such that every element is simple.
It can be shown the same way as in \cite[Prop. 3]{JePl} that in this case $E$  has an ordering
set of states. It follows that $E$ is strongly archimedean and $E\simeq V[0,u]$ for an order
unit space $(V,V^+,u)$, by Theorem \ref{th:archim}. Moreover, any element $v\in V$ is
simple, that is, it is a
linear combination $v=\sum_i c_i p_i$ for some sharp elements  $p_1,\dots,p_n$,
$\sum_ip_i=u$, \cite[Lemma 2]{JePl}.

\begin{theorem}\label{thm:contexts} Let $E$ be a convex SEA such that every element of $E$
is simple. Then $E$ is spectral. If $a=\oplus_i\mu_ip_i$ is a reduced representation of
$a$, then the spectral resolution of $a$ has the form $\{p_{a,\lambda}\}_{\lambda\in
[0,1]}$, where
\begin{align*}
p_{a,\lambda}= \oplus_{i=1}^{k-1} p_i,\quad  \lambda\in [\mu_{k-1},\mu_k),\qquad
k=1,\dots, n+1,
\end{align*}
where we put $\mu_0=0$, $\mu_{n+1}=1$ and $p_0=0$.

\end{theorem}

\begin{proof}
Let $a\in E$,  $a=\oplus_{i=1}^n \mu_ip_i$ be a reduced representation of $a$.
 Put $a^\circ:=\oplus_{i, \mu_i>0}\;  p_i$, then
$a^\circ$ is a projection cover of $a$. Indeed, it is clear that $a\le a^\circ\in P$ and if
$q\in P$ is such that $a\le q$, then $\mu_ip_i\le q$, so that $p_i\le q$ for all $i$ such
that $\mu_i>0$ (this follows from the fact that $q$ is principal), so that
$a^\circ=\vee_{i,\mu_i>0}\; p_i\le q$.

By \cite[Thm. 4.3 (i)]{Gucs}, we obtain that any $p_i$ is a \emph{function} of $a$,
which means that there are some elements $c_{i,0},\dots,c_{i,n-1}\in \mathbb R$ such that
\[
p_i=c_{i,0}1+\sum_{k=1}^{n-1}c_{i,k}a^k,
\]
where $a^k=a\circ\dots\circ a$. Here the sums are evaluated in the corresponding order unit
space $V$. It then follows from the axiom (S5) of the sequential
product that if $b\in E$ is such that $b|a$, then $b|a^k$ for all $k\ge 1$, so that for any projection
$p\in P$, we have $a\in C(p)$ if and only if $p_i\in C(p)$, $i=1,\dots,n$. Setting $B(a)$
to be the  Boolean subalgebra generated by
$p_1,\dots,p_n$, we see that $E$ has the b-property.

Let $b\in E$ be another element, with reduced representation $b=\oplus_j \lambda_j q_j$.
Then by the b-property proved above, we see that $aCb$ if
and only if $\{p_1,\dots,p_n\}\leftrightarrow \{q_1,\dots,q_m\}$. Hence there exist a
common refinement $\{r_k:=p_i\circ q_j\}$ and we may express both elements as
\[
a=\oplus_{k}\alpha_{k} r_k,\qquad b=\oplus_{k}\beta_{k} r_k,
\]
with $\alpha_{i,j}=\mu_i$ and $\beta_{i,j}=\lambda_j$. Note also that we have $r_k\in
P(a,b)$. Indeed, it is clear that $r_k\in PC(a,b)$, moreover, if $q$ commutes with both
$a$ and $b$, then we must have $q|p_i$ and $q|q_j$, so that $q|(p_i\circ q_j)$ by the
axiom  (S5). It is now easy to see that
\[
p:=\oplus_{k, \alpha_k<\beta_k}\;r_k\in P_\le (a,b).
\]
This shows the b-comparability property. The last statement on the spectral resolution
follows easily by the expression \eqref{eq:spectprojs} for the spectral projections.

\end{proof}

By uniqueness of the spectral resolutions, we obtain the following statement, cf.
\cite[Thm. 5.3]{FPspectres}.

\begin{corollary} Let $E$  be a convex SEA such that every element is simple. Then  every
 $a\in E$ has a unique reduced representation $a=\oplus_i \mu_ip_i$, moreover,  $\mu_i$
are precisely the eigenvalues of $a$.

\end{corollary}

\begin{corollary} Any context-spectral convex SEA is spectral.

\end{corollary}

\section*{Acknowledgements}

The research was supported by the grant VEGA 1/0142/20 and  the Slovak
Research and Development Agency grant APVV-20-0069.

\end{document}